\documentclass[conference]{IEEEtran}
\usepackage{graphicx}
\usepackage{cite}
\usepackage{amsmath}
\usepackage{amsfonts}
\usepackage{amssymb}
\usepackage{subfigure}
\usepackage{color}
\usepackage{amsthm}
\usepackage{enumerate}

\newcommand{\bm}{\mathbf} 
\newcommand{\be}{\begin{equation}}
\newcommand{\ee}{\end{equation}}
\newcommand{\bse}{\begin{subequations}}
\newcommand{\ese}{\end{subequations}}
\newcommand{\bea}{\begin{eqnarray}}
\newcommand{\eea}{\end{eqnarray}}

\newcommand{\ba}{{\bm a}}
\newcommand{\bb}{{\bm b}}

\newcommand{\bA}{{\bm A}}

\newcommand{\bD}{{\bf D}}

\newcommand{\bG}{{\bf G}}
\newcommand{\bH}{{\bf H}}

\newcommand{\bd}{{\bf d}}

\newcommand{\by}{{\bf y}}

\newcommand{\eye}{{\bm I}}
\newcommand{\I}{{\bm I }}

\newcommand{\BW}{{\boldsymbol{\mathcal W}}}

\newcommand{\bnu}{\mbox{\boldmath$\nu$}}

\theoremstyle{definition}
\newtheorem{proposition}{Proposition}

\graphicspath{{figures/}} 

\title{Prototype Filter Design for FBMC \\in Massive MIMO Channels \vspace{-0.0cm}}

\author{\normalsize Amir Aminjavaheri$^\dagger$, Arman Farhang$^*$, Linda E. Doyle$^*$, and Behrouz Farhang-Boroujeny$^\dagger$ \\
$^\dagger$ECE Department, University of Utah, Salt Lake City, Utah, USA, \\
$^*$CONNECT, The Telecommunications Research Centre, Trinity College Dublin, Ireland. \\
Email: \{aminjav, farhang\}@ece.utah.edu, \{farhanga, ledoyle\}@tcd.ie 

\thanks{This publication has emanated from research supported in part by a research grant from Science Foundation Ireland (SFI) and is co-funded under the European Regional Development Fund under Grant Number 13/RC/2077.}

\vspace{-0.0cm}}

\IEEEoverridecommandlockouts

\begin{document}

\maketitle

\begin{abstract}
We perform an asymptotic study on the performance of filter bank multicarrier (FBMC) in the context of massive multi-input multi-output (MIMO). We show that the signal-to-interference-plus-noise ratio (SINR) cannot grow unboundedly by increasing the number of base station (BS) antennas, and is upper bounded by a certain deterministic value. This is a result of the correlation between the multi-antenna combining tap values and the channel impulse responses between the terminals and the BS antennas. To solve this problem, we introduce a simple FBMC prototype filter design method that removes this correlation, enabling us to achieve arbitrarily large SINR values by increasing the number of BS antennas. 
\end{abstract}

\begin{IEEEkeywords} 
Massive MIMO, FBMC/OQAM, OFDM, SINR, channel equalization, asymptotic analysis. \vspace{-0.0cm}
\end{IEEEkeywords}

\section{Introduction}

Massive multiple-input multiple-output (MIMO) is one of the key technologies currently considered for the fifth generation (5G) of cellular networks. In a massive MIMO system, the base station (BS) is equipped with a large number of antennas, in the order of a hundred or a few hundreds, and is simultaneously serving tens of users. By increasing the number of BS antennas, the effects of uncorrelated noise and multiuser interference can be made arbitrarily small, \cite{marzetta2010noncooperative,ngo2013energy}, and hence, unprecedented network capacities can be achieved.

Due to its simplicity and robustness to multipath channels, orthogonal frequency division multiplexing (OFDM) is the dominant modulation format that is considered in the massive MIMO literature, \cite{marzetta2010noncooperative}, as well as most of the current wireless standards such as the 4G long term evolution (LTE). However, despite its many advantages, OFDM suffers from a number of drawbacks. In particular, due to the high side-lobe levels of the subcarriers, OFDM suffers from a large spectral leakage leading to high out-of-band emissions. Accordingly, stringent synchronization procedures are required in the uplink of multiuser networks. The users may experience different Doppler shifts, frequency offsets, timing offsets, etc., and maintaining the orthogonality between the subcarriers may not be possible without energy-consuming and resource-demanding procedures. Furthermore, utilization of non-contiguous spectrum chunks through carrier aggregation for the future high data rate applications is not possible in the uplink with OFDM as a result of high side-lobe levels of its subcarriers, \cite{iwamura2010carrier}. Moreover, to avoid interference, large guard bands are required between adjacent frequency channels, which in turn, lowers the spectral efficiency of OFDM. It should be emphasized that more strict requirements in terms of data rate, energy efficiency, and latency are defined for the 5G networks compared to the current ones in LTE, \cite{banelli2014modulation}. Therefore, the aforementioned shortcomings of OFDM and the requirements of 5G networks have stirred a great deal of interest in the area of waveform design among the research and industrial communities motivating introduction of alternative waveforms capable of keeping the advantages of OFDM while addressing its drawbacks, \cite{banelli2014modulation,Farhang2016,schaich2014waveform,farhang2014massive}.

Filter bank multicarrier (FBMC) is a 5G candidate waveform offering a significantly improved spectral properties over OFDM, by shaping the subcarriers using a prototype filter that is well-localized in both time and frequency, \cite{farhang2011ofdm}. Therefore, the uplink synchronization requirements can be significantly relaxed, \cite{aminjavaheri2015impact}, and carrier aggregation becomes a trivial task, \cite{perez2015mimo}. As a result of the above advantages, FBMC is currently being considered as an enabling technology in various research and industrial projects; see \cite{perez2015mimo} and the references therein.

The application of FBMC in massive MIMO channels has been recently studied in \cite{armanfarhang2014filter}, where its so-called self-equalization property leading to a channel flattening effect was reported through simulations. According to this property, the effects of channel distortions (i.e., intersymbol interference and intercarrier interference) will diminish by increasing the number of BS antennas. In \cite{aminjavaheri2015frequency}, multi-tap equalization is proposed for FBMC-based massive MIMO to improve the equalization accuracy compared to the single-tap equalization per subcarrier at the expense of a higher computational complexity. The authors in \cite{farhang2014pilot} show that the pilot contamination problem in multi-cellular massive MIMO networks, \cite{marzetta2010noncooperative}, can be resolved in a straightforward manner with FBMC signaling due to its special structure. These studies prove that FBMC is an appropriate match for massive MIMO and vice versa as they can both bring pivotal properties into the picture of 5G systems. Specifically, this combination is of a great importance as not only the same spectrum is being utilized by all the users but it is also used in a more efficient manner.

Since the literature on FBMC-based massive MIMO is not mature yet, these systems need to go through meticulous analysis and investigation. Hence, in this paper, we perform an in-depth analysis on the performance of FBMC in massive MIMO. We show that the self-equalization property shown through simulations and claimed in \cite{armanfarhang2014filter} and \cite{aminjavaheri2015frequency} is not very accurate. More specifically, by increasing the number of BS antennas, the channel distortions average out only up to a certain extent, but not completely. Thus, the SINR saturates at a certain deterministic level. This determines an upper bound for the SINR performance of the system. We derive an analytical expression for this saturation level, and propose a prototype filter design method to resolve the problem. With the proposed prototype filter in place, SINR grows without a bound by increasing the BS array size, and arbitrarily large SINR values are achievable.

It is worth mentioning that although the theories developed in this paper are applicable to all types of FBMC systems, the formulations are based on the most common type in the literature that was developed by Saltzberg, \cite{saltzberg1967performance}, and is known by different names including OFDM with offset quadrature amplitude modulation (OFDM/OQAM), FBMC/OQAM, and staggered multitone (SMT), \cite{farhang2011ofdm}. Throughout this paper, we refer to it as FBMC for simplicity.

The rest of the paper is organized as follows. To pave the way for the derivations presented in the paper, we review the FBMC principles in Section \ref{sec:system_model}. In Section \ref{sec:mmimo_fbmc}, we present the asymptotic equivalent channel model between the mobile terminals and the BS in an FBMC massive MIMO setup. This analysis will lead to an upper bound for the SINR performance of the system. Our proposed prototype filter design method is introduced in Section \ref{sec:prototype_modify}. The mathematical analysis of the paper as well as the efficacy of the proposed filter design technique are numerically evaluated in Section \ref{sec:numerical_results}. Finally, we conclude the paper in Section \ref{sec:conclusion}.

\textit{Notations:} Matrices, vectors and scalar quantities are denoted by boldface uppercase, boldface lowercase and normal letters, respectively. $[\bA]_{mn}$ represents the element in the $m^{\rm{th}}$ row and $n^{\rm{th}}$ column of $\bA$ and $\bA^{-1}$ signifies the inverse of $\bA$. $\I_M$ is the identity matrix of size $M\times M$. The superscripts $(\cdot)^{\rm T}$, $(\cdot)^{\rm H}$ and $(\cdot)^\ast$ indicate transpose, conjugate transpose, and conjugate operations, respectively. Also, $\ast$ represents the linear convolution, $\mathbb{E}\{\cdot\}$ denotes the expected value of a random variable, and $\Re\{\cdot\}$ signifies the real part of a complex number. The notation $\mathcal{CN}(0,\sigma^2)$ represents the circularly-symmetric complex normal distribution with zero mean and variance $\sigma^2$. Finally, $\delta_{ij}$ represents the Kronecker delta function.

\section{FBMC Principles} \label{sec:system_model}

We present the theory of FBMC in discrete time. Let $d_{m,n}$ denote the real-valued data symbol transmitted over the $m^{\rm th}$ subcarrier and $n^{\rm th}$ symbol time index. The total number of subcarriers is assumed to be $M$. To avoid the interference between the symbols and maintain the orthogonality, the data symbol $d_{m,n}$ should be phase adjusted using the phase term $e^{j\theta_{m,n}}$, where $\theta_{m,n} = \frac{\pi}{2}(m+n)$. Accordingly, each symbol has a $\pm\frac{\pi}{2}$ phase difference with its adjacent neighbors in time and frequency. The symbols are then pulse-shaped using the prototype filter $p(l)$, which has been designed such that $q(l) = p(l) \ast p^*(-l)$ is a Nyquist pulse with zero crossings at $M$ sample intervals. To express the above procedure in a
mathematical form, the discrete-time FBMC waveform can be written as, \cite{farhang2014filter},
\bea \label{eqn:fbmc_waveform}
x(l) = \sum_{n=-\infty}^{+\infty} \sum_{m=0}^{M-1} d_{m,n} a_{m,n}(l),
\eea
where 
\begin{align*}
a_{m,n}(l) &= p_m(l - nM/2) e^{j\theta_{m,n}}, 
\end{align*}
and $p_m(l) = p(l) e^{j\frac{2\pi ml}{M}}$ is the prototype filter modulated to the center frequency of subcarrier $m$. The functions $a_{m,n}(l)$ can be thought as a set of basis functions that are used to modulate the data symbols. Note that the spacing between successive symbols in the time domain is $M/2$ samples. In the frequency domain, the spacing between successive subcarriers is $1/M$ in normalized frequency. It can be shown that the basis functions $a_{m,n}(l)$ are orthogonal in the real domain, \cite{farhang2014filter}, i.e.,
\begin{align} \label{eqn:orthogonality}
\langle a_{m,n}(l), a_{m',n'}(l) \rangle_\Re &= \Re \bigg\{ \sum_{l = -\infty}^{+\infty} a_{m,n}(l) a^*_{m',n'}(l) \bigg\} \nonumber \\
&= \delta_{mm'} \delta_{nn'} .
\end{align}
Hence, the data symbols can be extracted from the synthesized signal, $x(l)$, according to
\be
d_{m,n} = \langle x(l), a_{m,n}(l) \rangle_\Re.
\ee
Fig.~\ref{fig:block_diagram} shows the block diagram of the FBMC transceiver. Note that considering the transmitter prototype filter $p(l)$, and the receiver prototype filter $p^*(-l)$, the overall effective pulse shape $q(l) = p(l) \ast p^*(-l)$ is a Nyquist pulse by design. Also, in practice, in order to implement the synthesis (transmitter side) and analysis (receiver side) filter banks efficiently, one can incorporate the polyphase implementation of filter banks to reduce the computational complexity, \cite{farhang2014filter}.

\begin{figure*}[!t]
\centering
\includegraphics[scale=0.85]{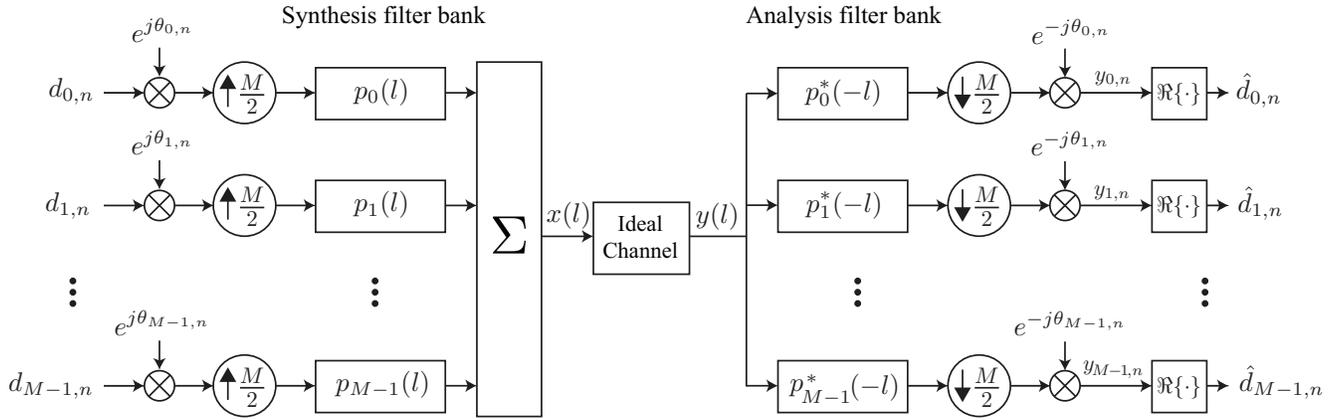}
\caption{Block diagram of the FBMC transceiver in discrete time. } \vspace{-0.2cm}
\label{fig:block_diagram}
\end{figure*}

The presence of a frequency-selective channel incurs some interference on the received symbols, and thus, one may adopt some sort of equalization to retrieve the transmitted symbols at the receiver side. Let $h(l)$ denote the impulse response of the channel. In this paper, we limit our study to a case where the channel impulse response remains time invariant over the interval of interest. Hence, the received signal at the receiver can be expressed as 
\begin{align} \label{eqn:y}
y(l) &= h(l) \ast x(l) + \nu(l) = \sum_{\ell=0}^{L-1} h(\ell) x(l-\ell) + \nu(l) ,
\end{align}
where $L$ is the length of the channel impulse response, and $\nu(l)$ is the additive noise.

At the receiver side, after matched filtering and phase compensation, and before taking the real part (see Fig.~\ref{fig:block_diagram}), the demodulated signal $y_{m,n}$ can be expressed as
\be \label{eqn:demod_symbol}
y_{m,n} = \sum_{n'=-\infty}^{+\infty} \sum_{m'=0}^{M-1} H_{mm',nn'} \hspace{1pt} d_{m',n'} + \nu_{m,n} ,
\ee
where $\nu_{m,n}$ is the noise contribution, and the interference coefficient $H_{mm',nn'}$ can be calculated according to 
\bse \label{eqn:siso_equiv_chan}
\begin{align} 
H_{mm',nn'} & =  h_{mm'}(n-n') \hspace{2pt} e^{j(\theta_{m',n'}-\theta_{m,n})},  \\
h_{mm'}(l) & = \big(p_{m'}(l) \ast h(l) \ast p_m^\ast(-l)\big)_{\downarrow \frac{M}{2}}. 
\end{align}
\ese
The symbol $\downarrow \frac{M}{2}$ denotes decimation with the rate of $\frac{M}{2}$. In (\ref{eqn:siso_equiv_chan}), $h_{mm'}(l)$ is the equivalent channel impulse response between the transmitted symbols at subcarrier $m'$ and the received ones at subcarrier $m$. This includes the effects of the transmitter pulse-shaping, the multipath channel, and the receiver pulse-shaping; see Fig.~\ref{fig:block_diagram}. According to (\ref{eqn:demod_symbol}), the demodulated symbol $y_{m,n}$ undergoes interference originating from other time-frequency symbols. In practice, the prototype filter $p(l)$ is designed to be well localized in time and frequency. As a result, the interference is limited to a small number of neighboring symbols around the desired time-frequency point $(m,n)$.

In order to devise a simple equalizer to combat the frequency-selective effect of the channel, it is usually assumed that the symbol period, $M/2$, is relatively large compared to the channel length, $L$. With this assumption, the demodulated signal $y_{m,n}$ can be expressed as, \cite{lele2008channel},
\be \label{eqn:siso_demod_symbol2}
y_{m,n} \approx H_{m} \big( d_{m,n} + u_{m,n} \big) + \nu_{m,n},
\ee
where $H_m \triangleq  \sum_{\ell=0}^{L-1} h(\ell)  e^{-j\frac{2\pi m\ell}{M}}$ is the channel frequency response at the center of the $m^{\rm th}$ subcarrier. The term $u_{m,n}$ is called the \emph{intrinsic interference} and is purely imaginary. This term represents the contribution of the intersymbol interference (ISI) and intercarrier interference (ICI) from the adjacent time-frequency symbols around the desired point $(m,n)$. Based on (\ref{eqn:siso_demod_symbol2}), the effect of channel distortions can be compensated using a single-tap equalizer per subcarrier. After equalization, what remains is the real-valued data symbol $d_{m,n}$, the imaginary term $u_{m,n}$, and the noise contribution. Finally, by taking the real part from the equalized symbol, one can remove the intrinsic interference and obtain an estimate of $d_{m,n}$.

It should be noted that the performance of the above single-tap equalization primarily depends on the validity of the assumption that the symbol duration is much larger than the channel length. However, in highly frequency-selective channels, where the above assumption is not accurate, more advanced equalization methods should be deployed to counteract the channel distortions, \cite{perez2015mimo}.

\section{Massive MIMO FBMC: Asymptotic Analysis} \label{sec:mmimo_fbmc}

In this section, we extend the formulation of the previous section to massive MIMO channels to be used in our subsequent asymptotic analysis. Then, we show that linear combining of the signals received at the BS antennas, using the channel frequency coefficients, leads to a residual interference even with an infinite number of BS antennas. Hence, the SINR is upper bounded by a certain deterministic value, and arbitrarily large SINR performances cannot be achieved as the number of BS antennas grows large. In the subsequent section, we show that this problem can be resolved through a simple prototype filter design method.

We consider a single-cell massive MIMO setup \cite{marzetta2010noncooperative}, with $K$ single-antenna mobile terminals (MTs) that are simultaneously communicating with a BS equipped with an array of $N$ antenna elements. As mentioned earlier, in this paper, we consider the uplink transmission while the results and our proposed technique are trivially applicable to the downlink transmission as well.

Let $x_k(l)$ represent the transmit signal of the terminal $k$. The received signal at the $i^{\rm th}$ BS antenna can be obtained as
\be \label{eqn:y_i}
y_i(l) = \sum_{k=0}^{K-1} x_k(l) \ast h_{i,k}(l) + \nu_i(l),
\ee
where $h_{i,k}(l)$ is the channel impulse response between the $k^{\rm th}$ terminal and the $i^{\rm th}$ BS antenna, and $\nu_i(l)$ is the additive noise at the input of the $i^{\rm th}$ BS antenna. We assume that the samples of the noise signal $\nu_i(l)$ are a set of independent and identically distributed (i.i.d.) $\mathcal{CN}(0,\sigma_\nu^2)$ random variables and the channel tap $h_{i,k}(l)$, $l \in \{0,\dots,L-1\}$, follows a $\mathcal{CN}(0,\rho(l))$ distribution.  Moreover, we assume that the channels corresponding to different terminals and different BS antennas are independent. Here, $\rho(l), l = 0,\dots,L-1$, is the channel power delay profile (PDP). Throughout this paper, we assume that the channel PDP is normalized such that $\sum_{l=0}^{L-1} \rho(l) = 1$. Moreover, we assume that for each terminal, the average transmitted power is equal to one, i.e., $\mathbb{E}\{ |x_k(l)|^2\} = 1$. As a result, considering the above channel model, the signal-to-noise ratio (SNR) at the input of the BS antennas can be calculated as $\text{SNR} = {1}/{\sigma_\nu^2}$. To simplify the analysis throughout the paper, we assume that the BS has a perfect knowledge of the channel state information (CSI).

Using (\ref{eqn:y_i}) and extending (\ref{eqn:demod_symbol}) to the MIMO case, we have
\be \label{eqn:mimo_demod_symbol1}
\by_{m,n} = \sum_{n'=-\infty}^{+\infty} \sum_{m'=0}^{M-1}  \bH_{mm',nn'} \hspace{2pt} \bd_{m',n'} + \bnu_{m,n},
\ee
where the $N \times 1$ vector $\by_{m,n}$ contains the demodulated symbols across different BS antennas and corresponding to the $(m,n)$ time-frequency point. The vector $\bnu_{m,n}$ contains the noise contributions across different BS antennas. $\bd_{m,n}$ contains the data symbols of all the MTs transmitted at the point $(m,n)$. $\bH_{mm',nn'}$ is an $N\times K$ matrix with its element $ik$, denoted by $H_{mm',nn'}^{i,k}$, representing the interference coefficient corresponding to the channel $h_{i,k}(l)$. The interference coefficient $H_{mm',nn'}^{i,k}$ can be calculated similar to (\ref{eqn:siso_equiv_chan}) as
\bse \label{eqn:mimo_equiv_chan}
\begin{align} 
H^{i,k}_{mm',nn'} & =  h^{i,k}_{mm'}(n-n') \hspace{2pt}  e^{j(\theta_{m',n'}-\theta_{m,n})},  \\
h^{i,k}_{mm'}(l) & = \big(p_{m'}(l) \ast h_{i,k}(l) \ast p_m^\ast(-l)\big)_{\downarrow \frac{M}{2}}. 
\end{align}
\ese

We assume that the BS utilizes a single-tap equalizer per subcarrier. Combining the elements of $\by_{m,n}$ through an $N\times K$ combining matrix $\BW_m$ and taking the real part of the resulting signal, the estimate of the transmitted data symbols for all the MTs can be obtained as
\begin{align} \label{eqn:mimo_est_d1}
\hat{\bd}_{m,n} &= \Re \left\{ \BW_m^{\rm H} \-\ \by_{m,n} \right\} \nonumber \\
&= \Re \Big\{ \sum_{n'=-\infty}^{+\infty} \sum_{m'=0}^{M-1} \BW_m^{\rm H} \bH_{mm',nn'} \bd_{m',n'}  + \BW_m^{\rm H} \bnu_{m,n} \Big\} \nonumber \\
&= \Re \Big\{ \sum_{n'=-\infty}^{+\infty} \sum_{m'=0}^{M-1}  \bG_{mm',nn'} \bd_{m',n'}  + \bnu'_{m,n} \Big\}  ,
\end{align} 
where $\bG_{mm',nn'} \triangleq \BW_m^{\rm H} \bH_{mm',nn'}$, and $\bnu'_{m,n} \triangleq \BW_m^{\rm H} \bnu_{m,n}$. In this paper, we consider three linear combiners, namely, maximum-ratio combining (MRC), zero-forcing (ZF), and minimum mean-square error (MMSE). These combiners can be obtained as, \cite{ngo2013energy},
\be \label{eqn:mrc_zf_mmse}
\BW_m = 
\begin{cases}
	\bH_m \bD_m^{-1} ,  & {\rm for~~ MRC}, \\
	\bH_m \left( \bH_m^{\rm H} \bH_m \right)^{-1},	   & {\rm for~~ ZF}, \\
	\bH_m \left( \bH_m^{\rm H} \bH_m + \sigma_\nu^2 \eye_K \right)^{-1},	   & {\rm for~~ MMSE}, 
\end{cases}
\ee
where $\bH_m$ is the matrix of channel coefficients at the center of $m^{\rm th}$ subcarrier, i.e., $\left[\bH_m\right]_{ik} = H_m^{i,k} \triangleq \sum_{l=0}^{L-1} h_{i,k}(l) e^{-j \frac{2\pi ml}{M}}$. In the MRC case, the $K\times K$ normalization matrix $\bD_m$ is a diagonal matrix that contains the squared norm of the $k^\mathrm{th}$ column of $\bH_m$ on its $k^{\rm th}$ diagonal element, i.e., $[\bD_m]_{kk} = \sum_{i=0}^{N-1} |H_m^{i,k}|^2$. Note that according to the law of large numbers, $\bD_m$ tends to $N \I_K$ as the number of BS antennas increases. In the following and to simplify the formulations, we only consider the case of MRC. We then show that the results are also applicable to the cases of ZF and MMSE as the number of BS antennas grows large.


Before we proceed, we review some results from probability theory. Let $\ba = [a_1,\dots,a_n]^{\rm T}$ and $\bb = [b_1,\dots,b_n]^{\rm T}$ be two random vectors each containing i.i.d. elements. Moreover, assume that $i^{\rm th}$ elements of $\ba$ and $\bb$ are correlated according to $\mathbb{E}\big\{ a_i^* b_i \big\} = C_{ab}$, $i = 1,\dots,n$. Consequently, according to the law of large numbers, the random variable $\frac{1}{n} \ba^{\rm H} \bb$ converges almost surely to $C_{ab}$ as $n$ tends to infinity.

In the asymptotic regime, i.e., as $N$ tends to infinity, the elements of $\bG_{mm',nn'} = \BW_m^{\rm H} \bH_{mm',nn'}$ can be calculated using the law of large numbers. Let $G^{kk'}_{mm',nn'}$ denote the element $kk'$ of $\bG_{mm',nn'}$. In the case of MRC, as $N$ grows large, $G_{mm',nn'}^{kk'}$ converges almost surely to
\be\label{eqn:Gkk'}
G^{kk'}_{mm',nn'} \rightarrow \mathbb{E} \Big\{ \left( H_{m}^{i,k} \right)^* H_{mm',nn'}^{i, k'} \Big\} .
\ee
To calculate the right hand side of (\ref{eqn:Gkk'}), we use (\ref{eqn:mimo_equiv_chan}) to find the equivalent channel impulse response between the transmitted data symbols and the received ones after combining the signals across different BS antennas. To this end, as $N$ grows large, the equivalent channel impulse response between the transmitted symbols at subcarrier $m'$  of MT $k'$ and the received ones at subcarrier $m$ of MT $k$ tends to\footnote{Note that in (\ref{eqn:Gkk'}) and (\ref{eqn:asym_equiv_chan}), we have used the letters $G$ and $g$, respectively, to denote the equivalent channel coefficients \emph{after combining}. On the other hand, letters $H$ and $h$ have been used in (\ref{eqn:mimo_equiv_chan}), to refer to the respective channel coefficients \emph{before combining}.}
\begin{align} \label{eqn:asym_equiv_chan}
&g_{mm'}^{kk'}(l) \rightarrow \mathbb{E} \Big\{ \left( H_{m}^{i,k} \right)^* \Big(p_{m'}(l) \ast h_{i, k'}(l) \ast p_m^\ast(-l)\Big)_{\downarrow \frac{M}{2}} \Big\} \nonumber \\
&\hspace{10pt}=  \Big( p_{m'}(l) \ast \mathbb{E} \Big\{ \left( H_{m}^{i,k} \right)^* h_{i,k'}(l) \Big\} \ast p_m^\ast(-l)\Big)_{\downarrow \frac{M}{2}} .
\end{align}
The above expression includes a correlation between the channel frequency coefficient $H_m^{i,k}$ and the channel impulse response $h_{i,k'}(l)$. This correlation can be calculated as
\begin{align} \label{eqn:channel_corr}
\mathbb{E} \Big\{ \left( H_{m}^{i,k} \right)^* h_{i, k'}(l) \Big\} &= \sum_{\ell = 0}^{L-1} \mathbb{E} \left\{ h_{i,k}^*(\ell) h_{i,k'}(l) \right\} e^{j\frac{2 \pi \ell m}{M}}  \nonumber \\
&= \rho(l) e^{j\frac{2 \pi l m}{M}} \delta_{kk'} = \rho_m(l) \delta_{kk'} ,
\end{align}
where $\rho_m(l) \triangleq \rho(l) e^{j\frac{2 \pi l m}{M}}$ .

\begin{proposition} \label{prp:saturation1}
In an FBMC-based massive MIMO system, as the number of BS antennas tends to infinity, the effects of multiuser interference and noise vanish. However, some residual ISI and ICI from the same user remain. In particular, for a given user $k$, the equivalent channel impulse response between the transmitted data symbols at subcarrier $m'$ and the received ones at subcarrier $m$ tends to
\be \label{eqn:MRC_equiv_response}
g^{kk}_{mm'}(l) \rightarrow \Big(p_{m'}(l) \ast \rho_m(l) \ast p_m^\ast(-l)\Big)_{\downarrow \frac{M}{2}} .
\ee
As a result, the SINR saturates to 
\be \label{eqn:SINR_sat}
{\rm SINR}^k_{m,n} \rightarrow \frac{ \Re^2 \big\{ G^{kk}_{mm,nn} \big\} }{ \mathop{\sum\limits_{n'=-\infty}^{+\infty} \sum\limits_{m'=0}^{M-1}}\limits_{(m',n')\neq(m,n)} \Re^2 \big\{ {G}^{kk}_{mm',nn'} \big\}} ,
\ee
where $G^{kk}_{mm',nn'} = {g}^{kk}_{mm'}(n-n') \hspace{1pt} e^{j(\theta_{m',n'}-\theta_{m,n})}$. 
\end{proposition}
\begin{proof}
As suggested by (\ref{eqn:channel_corr}), when $k'\neq k$, the channel response tends to zero. Hence, multiuser interference fades away. A similar argument can be developed for the noise contribution. However, when $k' = k$, which implies the interference from the same user on itself, the channel response tends to (\ref{eqn:MRC_equiv_response}). Notice that due to the presence of $\rho_m(l)$, the orthogonality condition of (\ref{eqn:orthogonality}) does not hold anymore even with an infinite number of BS antennas. Hence, some residual ISI and ICI will remain and will cause the SINR to saturate at the level in (\ref{eqn:SINR_sat}). 
\end{proof}

Although the above discussions and analysis was made for MRC, we note that Proposition \ref{prp:saturation1} is valid for the ZF and MMSE combiners as well. In particular, for the ZF and MMSE combiners, one may use the fact that due to the law of large numbers, when $N$ grows large, $\frac{1}{N} \bH_m^{\rm H} \bH_m$ tends to $\eye_K$, and hence the ZF and MMSE matrices in (\ref{eqn:mrc_zf_mmse}) tend to that of the MRC, \cite{ngo2013energy}. Thus, the same asymptotic SINR value and channel impulse response as for the MRC can be obtained for the ZF and MMSE combiners.

\section{Proposed Prototype Filter Design Method} \label{sec:prototype_modify}

As discussed in the previous section, even with an infinite number of BS antennas, some residual ICI and ISI remain due to the correlation between the combining tap values and the channel impulse responses between the MTs and the BS antennas. As a solution to this problem, in this section, we propose a prototype filter design method to remove the above correlation.

In (\ref{eqn:MRC_equiv_response}), the problematic term that leads to the saturation issue is the modulated channel PDP, $\rho_m(l)$. In the absence of this term, the channel response $g^{kk}_{mm'}(l) = \big( p_{m'}(l) \ast p^*_m(-l) \big)_{\downarrow \frac{M}{2}}$ does not incur any interference and the orthogonality condition is completely satisfied, provided that $q(l) = p(l) \ast p^*(-l)$ is a Nyquist pulse. This observation suggests that we can modify the prototype filter used at the BS such that 
\be \label{eqn:q_tilde}
q(l) = p(l) \ast \rho(l) \ast \tilde{p}^*(-l) ,
\ee
is still a Nyquist pulse. In (\ref{eqn:q_tilde}), $\tilde{p}(l)$ denotes the modified prototype filter. Applying a discrete-time Fourier transform (DTFT) to (\ref{eqn:q_tilde}), we have
\be
Q(\omega) = P(\omega) \bar{\rho}(\omega) \tilde{P}^\ast(\omega) ,
\ee
where $\bar{\rho}(\omega)$ denotes the DTFT of $\rho(l)$. We note that since $q(l) = p(l) \ast p^*(-l)$, we may write $Q(\omega) = | P(\omega) |^2$. Thus, 
\be \label{eqn:modified_P}
\tilde{P}(\omega) = \frac{P(\omega)}{\bar{\rho}^\ast (\omega)} .
\ee
Finally, applying an inverse-DTFT to $\tilde{P}(\omega)$ will give us the impulse response of the modified prototype filter, $\tilde{p}(l)$. The following proposition summarizes the above results.

\begin{proposition} \label{prp:proto_modif}
The SINR saturation problem can be resolved by incorporating the modified prototype filter $\tilde{P}(\omega) = \frac{P(\omega)}{\bar{\rho}^\ast (\omega)}$ at the BS. Consequently, as $N$ grows large, ISI and ICI in addition to the effects of multiuser interference and noise tend to zero, and arbitrarily large SINR values can be achieved.
\end{proposition}
\begin{proof}
Following (\ref{eqn:modified_P}), the equivalent channel impulse response in (\ref{eqn:MRC_equiv_response}) tends to that of an ideal channel. Hence, the effects of ICI and ISI will vanish asymptotically. Note that since the channels of different users are independent (see (\ref{eqn:channel_corr})), the effect of multiuser interference still tends to zero with the modified prototype filter in place. A similar argument applies for the noise contribution.
\end{proof}

It is worth to mention a number of points here. First, we note that in the above approach, only the prototype filter used at the BS is modified and other parts of the FBMC transceiver, including the combining taps, will remain unchanged. Also, it should be noted that according to (\ref{eqn:modified_P}), the modified prototype filter depends on the channel PDP. Hence, the BS needs to estimate the channel PDP to be able to construct $\tilde{p}(l)$. Fortunately, in massive MIMO scenarios, the problem of channel PDP estimation is relatively easy and feasible. In particular, the channel PDP can be determined by calculating the variance of channel impulse responses across different BS antennas. As the number of BS antennas increases, according to the law of large numbers, this estimate becomes closer to the exact channel PDP. Last but not least, we note that in the above analysis, we did not make any assumption about the flatness of the channel response over the bandwidth of the subcarriers. Thus, the result obtained in Proposition \ref{prp:proto_modif} is valid for any frequency-selective channel.

\section{Numerical Results} \label{sec:numerical_results}

In this section, we evaluate the analysis of the previous sections as well as the efficacy of our proposed prototype filter design method using computer simulations. We let $M=256$ and assume there are $K=10$ terminals in the network. The terminals use the PHYDYAS prototype filter, \cite{bellanger2010fbmc}, with overlapping factor of 4, to synthesize their FBMC signals. A normalized exponentially decaying channel PDP, $\rho(l) = e^{-\alpha l} / \big( \sum_{\ell=0}^{L-1} e^{-\alpha \ell} \big), l=0\dots,L-1$, with $\alpha = 0.1$ and $L=40$ is assumed. At the BS side, a modified prototype filter designed according to (\ref{eqn:modified_P}) is used to analyze the received FBMC signals across different antennas. 

Figs.~\ref{fig:modified_P_time} and \ref{fig:modified_P_freq} present the time and frequency responses, respectively, of the modified prototype filter and compare them against the original PHYDYAS filter. Moreover, the sinc pulse, as the pulse-shape of the subcarriers in OFDM, is shown in Fig.~\ref{fig:modified_P_freq} as a reference. As shown, both  prototype filters provide a significantly lower spectral leakage compared to the sinc pulse. It should be mentioned that although the original and modified filters do not differ significantly in shape, they lead to completely different SINR behaviors, as it is shown in the following. 

\begin{figure}[!t]
\centering
\includegraphics[scale=0.51]{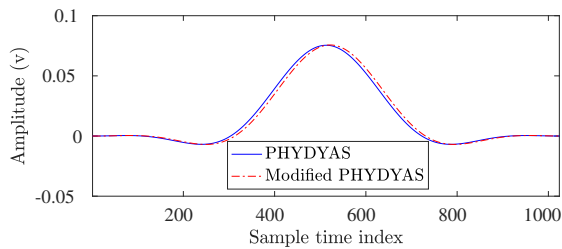} \vspace{-0.2cm}
\caption{Impulse responses of the PHYDYAS and modified PHYDYAS filters. }
\label{fig:modified_P_time}
\end{figure}

\begin{figure}[!t]
\centering
\includegraphics[scale=0.51]{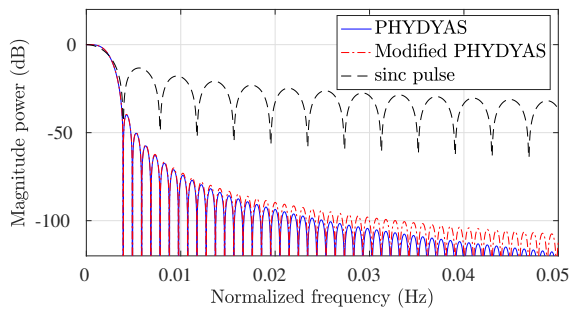} \vspace{-0.2cm}
\caption{Frequency responses of the PHYDYAS and modified PHYDYAS filters and comparison with the sinc pulse. } \vspace{-0.12cm}
\label{fig:modified_P_freq}
\end{figure}

\begin{figure}[!t]
\centering
\includegraphics[scale=0.53]{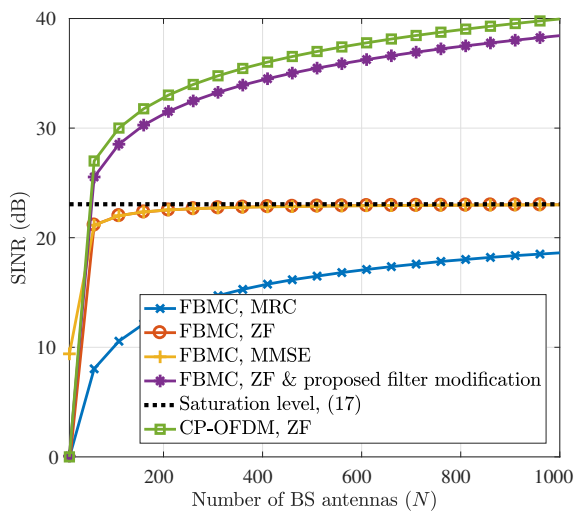} \vspace{-0.2cm}
\caption{SINR performance comparison.} \vspace{-0.4cm}
\label{fig:sinr}
\end{figure}

We next compare the SINR performance of the FBMC transmission with and without prototype filter modification. Fig.~\ref{fig:sinr} shows the average SINR (with averaging over different channel realizations) versus the number of BS antennas. The noise level is selected such that the SNR at the input of the BS antennas is equal to $10$~dB. From Fig.~\ref{fig:sinr} we can see that when the prototype filter is not modified, the SINR performance of all three detectors, i.e., MRC, ZF, and MMSE, tends to the saturation level predicted by (\ref{eqn:SINR_sat}) as $N$ grows large. However, when we incorporate the modified prototype filter, the SINR grows without a limit by increasing $N$. Here, only the case of ZF detector is shown. Also, the SINR performance of OFDM with cyclic prefix (CP-OFDM) and with ZF detector is shown as a benchmark. There is a small difference (around $1.5$ dB) between the SINR of CP-OFDM and FBMC with our proposed modified prototype filter. This is due to the fact that the presence of CP in OFDM leads to a complete removal of all various interference components. In contrast, the FBMC waveform is designed to increase the bandwidth efficiency, by not including any CP overhead and providing much lower out-of-band emission.

\section{Conclusion} \label{sec:conclusion}

In this paper, we studied the performance of FBMC transmission in the context of massive MIMO. We considered single-tap equalization per subcarrier using the conventional linear combiners, i.e., MRC, ZF, and MMSE. One of our findings in this paper was that the correlation between the combining tap values and the channel impulse responses leads to an interference which does not fade away as the BS array size increases. Therefore, the SINR is upper-bounded by a certain deterministic value and arbitrarily large SINR values cannot be achieved. We derived an analytical expression for this upper bound, identified the source of SINR saturation, and proposed a prototype filter design method to remove the above correlation and resolve the problem.

\bibliographystyle{IEEEtran} 
\bibliography{IEEEabrv,MM-FBMC}

\end{document}